\documentclass[10pt,a4paper]{article}
\usepackage[utf8]{inputenc}
\usepackage[english]{babel}
\usepackage{amsmath}
\usepackage{amsfonts}
\usepackage{amssymb}
\usepackage{amsthm}
\usepackage{cite}
\usepackage{graphicx}
\usepackage[left=2cm,right=2cm,top=2cm,bottom=2cm]{geometry}
\newtheorem{te}{Theorem}

\newtheorem{lem}{Lemma}

\newtheorem{de}{Definition}
\newtheorem{ex}{Example}

\author{R. Azuaje and A. M. Escobar-Ruiz\\
Departamento de F\'{i}sica, Universidad Aut\'onoma Metropolitana Unidad Iztapalapa,\\ San Rafael Atlixco 186, 09340, Ciudad de M\'exico, M\'exico}
\title{On particular integrability for (co)symplectic and (co)contact Hamiltonian systems}
\begin{document}
\maketitle

\begin{abstract}
As a generalization and extension of our previous paper [Escobar-Ruiz and Azuaje, J. Phys. A: Math. Theor. {\bf 57}, 105202 (2024)], in this work, the notions of particular integral and particular integrability in classical mechanics are extended to the formalisms of cosymplectic, contact and cocontact geometries. This represents a natural scheme to study nonintegrable time-dependent systems where only a part of the whole dynamics satisfies the conditions for integrability. Specifically, for Hamiltonian systems on cosymplectic, contact and cocontact manifolds, it is demonstrated that the existence of a particular integral allows us to find certain integral curves from a reduced, lower dimensional, set of Hamilton's equations. In the case of particular integrability, these trajectories can be obtained by quadratures. Notably, for dissipative systems described by contact geometry, a particular integral can be viewed as a generalization of the important concept of dissipated quantity as well. 

\end{abstract}

\section{Introduction}
\label{sec1}

In a classical mechanical system, the existence of conserved quantities reduces the difficulty of solving the equations of motion. Also, it is common knowledge that in the symplectic framework the Liouville integrals of motion (conserved quantities) for a conservative Hamiltonian system are intimately connected to symmetries \cite{AMRC2008,AKN2006,Roman2020}. The famous Liouville theorem on integrability of Hamiltonian systems \cite{AMRC2008,AKN2006,BB98}, with $n$ degrees of freedom, states that the knowledge of $n$ functionally independent symmetries (constants of motion, Liouville first integrals) in involution enables us to solve the corresponding Hamilton's equations using quadratures, i.e., by performing a finite  number of algebraic operations, taking inverse functions, and possibly calculating the integrals of explicit known functions.

In Physics, for a given set of initial data in the corresponding domain, a constant of motion is viewed as a conserved quantity during the time evolution. For this reason, they can also be called global integrals. The concept of particular integral introduced in \cite{Turbiner2013} generalizes the notion of a global integral. In simple terms, a particular integral is a conserved quantity for possibly certain (sub)sets of initial conditions only. One important rationale to investigate particular integrals is that they allow us to study and \textit{solve} non-integrable systems in certain points and regions where the dynamics of the system satisfies the regularity requirements for integrability which, in turn, leads to what we call particular integrability. Accordingly, in \cite{escobar2023particular} the concepts of particular integral and particular integrability were introduced within the formalism of symplectic geometry, namely, they were considered for Hamiltonian systems defined on symplectic manifolds. Symplectic geometry very often is regarded as the natural geometric formalism to develop the Hamiltonian theory of classical mechanics. Nevertheless, only autonomous (conservative) systems can be described under this framework. This significant limitation justifies the present attempt to extend the previous study \cite{escobar2023particular} to the case of non-autonomous systems as well as the dissipative ones.

Other geometric constructions such as cosymplectic, contact as well as cocontact geometry admit alternative formulations of Hamiltonian mechanics. They have been proven to be instrumental in describing non-autonomous and dissipative mechanical systems within the Hamiltonian theory. Currently, contact and cocontact Hamiltonian mechanics are subjects of intense active research, see for example \cite{LL2020,BG2021,GLR2022,AE2023,azuaje2024lie}. For instance, in the case of dissipative systems, the so called dissipated quantities were studied in relation with Noether symmetries \cite{BG2021,GGMRR2020,LL2020}, giving as a result dissipation laws analogous to the conservation theorems occurring in conservative systems. In fact, we will show that a dissipated quantity is a special case of a particular integral. 

In \cite{escobar2023particular}, the particular integrability has been presented as a more general notion than Liouville integrability. Hence, they share important features. For example, a common key element is the property of functional independence of certain functions defined on the relevant manifold. It is worth mentioning this property and the corresponding process of reducing the dimension of the manifold: let $M$ be a smooth real manifold, it is said that the functions $g_{1},\cdots,g_{s}\in C^{\infty}(M)$ with $s\leq dim(M)$ are functionally independent if for any $p$ in a dense subset of $M$ the differential maps $dg_{1}|_{p},\cdots,dg_{s}|_{p}$ are linearly independent. In other words, $p$ is a regular point of the differentiable function $G:M\longrightarrow \mathbb{R}^{s}$ defined by $G=(g_{1},\cdots,g_{s})$. Next, we consider the level set $M_{g}=\lbrace x\in M: g_{i}(x)=\alpha_{i}, \alpha_{i}\in \mathbb{R}\rbrace$, and assuming the $g_{1},\ldots,g_{s}$ are functionally independent. Now, using the submersion level set theorem \cite{Lee2012} we obtain that $M_{g}$ is a smooth submanifold of $M$ of dimension $dim(M_g)=dim(M)-s$. In general, this reduction is beyond the standard separation of variables. 

The aim of this paper is to introduce the notions of particular integral and particular integrability in Hamiltonian mechanics within the theory of cosymplectic, contact and cocontact
geometry in detail. The associated reduction of the equations of motion, using particular integrals, is investigated in each case. We emphasize that the nature of the paper is conceptual, the main results are established in the form of Theorems and concrete examples to illustrate the key ideas are provided. 

This paper is structured in the following manner. In section \ref{sec2}, we present the notion of particular integral for a classical mechanical system defined by a smooth vector field on a smooth generic real manifold; in this section we emphasize the similarities and differences between global constants of motion (also called integrals of motion or first Liouville integrals) and particular integrals. Afterwards, in section \ref{sec3} we describe the notions of particular integral and particular integrability for Hamiltonian systems defined on a smooth manifold equipped with different geometric structures; this section is divided into three parts regarding the cosymplectic, contact and cocontact geometry, respectively. In the contact case, for a physically relevant dissipative Hamiltonian system, it is shown that a dissipated quantity is a special case of a particular integral explicitly, and the corresponding reduction of the equations of motion is carried out.

\section{Particular integrals for classical mechanical systems}
\label{sec2}

Classical Newtonian mechanics is focused to describe the time evolution (motion) of physical systems whose future and past are unambiguously determined by the initial data, i.e., the set of initial positions and velocities of all particles (points) forming the system; such systems are called (classical) mechanical systems. For a mechanical system, the phase space is the set whose elements are the sets of positions and velocities of all points of the given system \cite{Arnold92}. Remarkably, Poincaré visualized a mechanical system as a special vector field on phase space, where a trajectory is a smooth curve tangent at each of its points to the vector based at that point; which lead him to the important result of a smooth manifold as the phase space in classical mechanics \cite{AMRC2008}. 

Let $M$ be a smooth real manifold of dimension $n$ and $V$ a smooth vector field on $M$. $V$ defines a (classical) mechanical system on $M$ whose trajectories are the integral curves of $V$ \cite{AMR88}. Let us remember that an integral curve of $V$ is a curve $\gamma:I\subset \mathbb{R}\longrightarrow M$ with the property $\frac{d}{dt}\gamma(t)=V_{\gamma(t)}$ $\forall t\in I$. Let $(x^{1},\cdots,x^{n})$ be local coordinates on $M$, the local expression of $V$ is $V=V^{1}(x^{1},\cdots,x^{n})\frac{\partial}{\partial x^{1}}+\cdots+V^{n}(x^{1},\cdots,x^{n})\frac{\partial}{\partial x^{n}}$. Also, the integral curves $\gamma(t)=(x^{1}(t),\cdots,x^{n}(t))$ of $V$ satisfy the system of $n$ first order ordinary differential equations (ODE's)
\begin{equation}
\label{eqmot}
\left\lbrace \begin{array}{c}
\frac{d}{dt}{x}^{1} \,=\, V^{1}(x^{1},x^2\cdots,x^{n}),\\ 
\vdots \\ 
\frac{d}{dt}{x}^{n}\,=\, V^{n}(x^{1},x^2,\cdots,x^{n}).
\end{array} \right.
\end{equation}
Take a mechanical system $(M,V)$, where $M$ denotes the corresponding phase space and $V$ is the dynamical vector field. The system of first order $n$ differential equations (\ref{eqmot}) is the corresponding system of (local) equations of motion (the motion of a system in classical mechanics is described using ODE's \cite{Arnold92}). The time evolution of any observable $g=g(x)\in C^{\infty}(M)$ along the integral curves of the system reads
\begin{equation}
\label{eqevo}
\frac{d}{dt}g(\gamma(t))=dg_{\gamma(t)}(\frac{d}{dt}\gamma(t))=(\frac{d}{dt}\gamma(t))(g)=V_{\gamma(t)}(f)=(Vg)(\gamma(t))\ ,
\end{equation}
so the temporal evolution of $g$ is simply given by
\begin{equation}
\dot{g}\ = \ Vg \ .
\end{equation}
The function $g$ is called a conserved quantity (or an integral of motion) for the system $(M,V)$ if $g$ remains constant on its integral curves, which is equivalent to the condition $Vg=0$. The occurrence of a constant of motion enable us to derive a (reduced) mechanical system whose dynamics is contained in the original (bigger) system, i.e., we can construct a reduced system such that its trajectories are also integral curves of the original one. Thus, if the system admits a constant of motion one can look for integral curves that are (locally) solutions of a reduced system of dynamical equations. Indeed, let us suppose that $g$ is an integral of motion of the system $(M,V)$, we have that the level set 
\[ 
M_{g} \ = \ \lbrace\, x\in M: \ g(x)\,=\,\alpha, \, \alpha\in\mathbb{R}\rbrace \ ,
\]
defines a (smooth) submanifold of $M$ of codimension $1$ for regular values of $g$ \cite{Lee2012}. Moreover, it is invariant under the dynamics of $V$. Explicitly, if $\gamma:I\longrightarrow M$ is an (integral) curve of $V$ such that $\gamma(t_{0})\in M_{g}$ for a given $t_{0}\in I$, also   $\gamma(t)\in M_{g}$ at any $t\in I$. We can prove this last statement as follows: let us suppose that $\gamma:I\longrightarrow M$ is an integral curve of $V$ with the property that $\gamma(t_{0})\in M_{g}$ at $t_{0}\in I$, hence $g(\gamma(t_{0}))=\alpha$ and since $g$ is a conserved quantity it is constant on the whole trajectory, i.e., $g(\gamma(t))=\alpha$ for every $t\in I$, so we conclude that $\gamma(t)\in M_{g}$ for every $t\in I$ as well. Now we can look for trajectories of the system $(M,V)$ that live in $M_{g}$. 

It is important remarking that the set of the trajectories of the system $(M,V)$ that live in $M_{g}$ is not the whole set of the trajectories of the system, but some element of this first bigger set are simply determined by a (reduced) system of ODE's. Indeed, if $g$ is an integral of motion of $(M,V)$, the integral curves that live in $M_{g}$ are (locally) solutions of a system of $n-1$ differential equations: the trajectories of $V$ that live in $M_{g}$ are integral curves of the restriction $V|_{M_{g}}$ of the vector field $V$ to $M_{g}$, so if $(y^{1},\cdots,y^{n-1})$ are local coordinates on $M_{g}$ then such integral curves are solutions of the differential equation in vectorial form $\dot{Y}=V|_{M_{g}}(Y)$, where $Y=(y^{1},\cdots,y^{n-1})$, which is identical to a set of $n-1$ differential equations.

Now we present the concept of particular integral explaining its relation with a constant of motion.
\begin{de}
We call $p\in C^{\infty}(M)$ a particular integral of motion for $(M,V)$ if the equation $Vp=a\,p$ is satisfied for some function $a\in C^{\infty}(M)$.
\end{de}
In the case $a=0$ we have a constant of motion. Similarly to the previous case where a constant of motion is involved, the occurrence of a particular integral allows us to look for certain trajectories of the system that are (locally) solutions of a (reduced) set of differential equations of motion. We have the following result.
\begin{lem}
If $p$ is a particular integral of $(M,V)$ it follows that the level set $M_{p}$ is invariant (closed) under the dynamics of $V$.
\end{lem}
The proof of the previous lemma is completely analogous to the one presented in \cite{escobar2023particular} for the symplectic framework. For further details we refer the reader to \cite{escobar2023particular}.  

Now, given a particular integral we can look for the trajectories of the system $(M,V)$ that live in $M_{p}$, and like for a constant of motion, if $p$ is a particular integral of $(M,V)$ then the integral curves that live in $M_{p}$ are (locally) solutions of a (reduced) system of ODE's, $n-1$ equations.

Up to this point general classical mechanical systems have been considered. We are interested in Hamiltonian systems, these are mechanical systems on smooth (real) manifolds where the dynamical vector fields are the Hamiltonian vector fields for distinguished functions called the Hamiltonian functions. The construction of Hamiltonian vector fields on a smooth manifold requires the existence of an additional geometric structure defining a Poisson structure or a more general structure called a Jacobi structure (these last structures are not common in the literature and are less known than the Poisson ones). Given a smooth real manifold equipped with a Poisson structure or a Jacobi structure, the corresponding Hamiltonian vector field for a given function is assigned according to a one-to-one correspondence, between vector fields and 1-forms on such manifold, defined by means of the geometric structure. Generally Poisson structures are defined by symplectic structures, in fact, symplectic geometry is considered the natural geometric framework where the theory of Hamiltonian mechanics is developed (see \cite{AMRC2008,LR89,Torres2020,Lee2012,marsden2013introduction}). Nevertheless, only autonomous conservative systems can be described under this formalism. Other geometric frameworks allow alternative formulations of Hamiltonian mechanics. They are instrumental to describe physically relevant non-autonomous and dissipative mechanical systems as Hamiltonian systems. These formalisms are cosymplectic geometry, contact geometry and cocontact geometry. For details on these geometric formalisms and the associated formulations of Hamiltonian mechanics see \cite{LR89,CLL92,LS2017,LL2019,LS2017,BCT2017,Bravetti2017,Letal2022}.

\section{Hamiltonian systems: particular integrals and particular integrability }
\label{sec3}

In this section the geometric notion of particular integral is introduced for Hamiltonian systems on cosymplectic, contact and cocontact manifolds; we employ the same notation as in \cite{AE2023,azuaje2024lie} where a brief review of the geometric frameworks and the associated formulations of Hamiltonian mechanics is given. Afterwards, the concept of particular integrability is studied in each case. Analogous results to those given in \cite{escobar2023particular} for the symplectic case, are presented in detail.

\subsection{Cosymplectic Hamiltonian systems}
\label{subcosymplectic}

Let $(M,\Omega,\eta,H)$ be a cosymplectic Hamiltonian system with $n$ degrees of freedom (${\rm dim}(M)=2n+1$). In order to introduce the notion of particular integrability in the cosymplectic formalism, we follow the spirit of the original notion, i.e., a particular integral $g$ must be characterized by the condition $\lbrace g,H\rbrace=ag$ for some function $a\in C^{\infty}(M)$. So we consider a particular integral of $(M,\Omega,\eta,H)$ as a function $g\in C^{\infty}(M)$ such that $X_{H}g=ag$ for some function $a\in C^{\infty}(M)$; in terms of the Poisson bracket we have that $g$ is a particular integral if and only if $\lbrace g,H\rbrace=ag$. It is important to remark that in the cosymplectic framework, the dynamical vector field of the Hamiltonian system $(M,\Omega,\eta,H)$ is given by $E_{H}=X_{H}+R$ ($R$ being the Reeb vector field), so in order to be able to look for the trajectories of the system that live in $M_{g}=\lbrace\, x\in M: \ g(x)\,=\,0\,\rbrace$ we must have $E_{H}g|_{g=0}=0$, which is fulfilled by requiring that $Rg=0$ ($g$ is time-independent). In conclusion, we propose the following definition.
\begin{de}
A particular integral of $(M,\Omega,\eta,H)$ is a function $g\in C^{\infty}(M)$ which satisfies $Rg=0$ and $X_{H}g=ag$ for a function $a\in C^{\infty}(M)$.
\end{de}
Assuming $g\in C^{\infty}(M)$ is a particular integral for $(M,\Omega,\eta,H)$, one can look for the trajectories of the system that live in $M_{g}=\lbrace\, x\in M: \ g(x)\,=\,0\,\rbrace $. Provided that $0$ is a regular value of $g$, $M_{g}$ corresponds to a (smooth) submanifold of dimension $(2n+1)-1=2n$. In canonical coordinates $(q^{1},\cdots,q^{n},p_{1},\cdots,p_{n},t)$ on $(M,\Omega,\eta)$, the trajectories of the system are the solutions of the Hamilton's equations of motion 
\begin{equation}
\dot{q^{i}} =\frac{\partial H}{\partial p_{i}}, \hspace{1cm}
\dot{p_{i}} =-\frac{\partial H}{\partial q^{i}}\qquad ;\qquad i=1,2,3,\ldots,n \ .
\end{equation}
The construction of a reduced Hamiltonian system by mean of a particular integral in the cosymplectic case is completely analogous to the one in the symplectic case \cite{escobar2023particular}, namely, we can solve the original system of equations of motion by solving first a reduced system of Hamilton equations of motion and afterwards integrating a differential equation.

For cosymplectic Hamiltonian systems the notion of Liouville integrability is completely analogous to the one in the symplectic case, namely, for a cosymplectic Hamiltonian system with $n$ degrees of freedom, the existence of $n$ functionally independent integrals of motion in involution allows to find the solutions of the Hamilton equations of motion by quadratures \cite{GMS2002,JL2023,LMV2011}. Now we introduce the notion of particular integrability in the framework of cosymplectic geometry. Observe that the concept of particular involution only involve the Poisson bracket, so in the cosymplectic case we have the same notion as in the symplectic case, i.e., the functions  $g_{1},\cdots,g_{s}\in C^{\infty}(M)$ are in particular involution if they satisfy the condition 
\begin{equation}
\label{eqPS}
\lbrace g_{i},g_{j}\rbrace=a^{ij}_{1}g_{1}+\cdots+a^{ij}_{s}g_{s}
\end{equation}
for some functions $a^{ij}_{1},\ldots,a^{ij}_{s}\in C^{\infty}(M)$. Eventually, we arrive to the following theorem analogous to the theorem on particular integrability given in \cite{escobar2023particular} (theorem 2) for the symplectic case.
\begin{te}
\label{tePI2}
Let $g_{1},\cdots,g_{n}\in C^{\infty}(M)$ be a set of  functionally independent elements in particular involution such that $R\,g_{i}=0$. If for each $i\in \lbrace 1,\cdots,n\rbrace$ $\lbrace g_{i},H\rbrace=b^{i1}g_{1}+\cdots+b^{in}g_{n}$ for some $b^{ij}\in C^{\infty}(M)$, then the trajectories of the Hamiltonian system $(M,\Omega,\eta,H)$ lying in $M_{g}=\lbrace x\in M: \ g_{1}(x)=\cdots=g_{n}(x)=0\rbrace$ can be obtained using quadratures. 
\end{te}
Observe that $g_{1},\cdots,g_{n}$ in theorem \ref{tePI2} may be either particular integrals or even constants of motion in particular involution, of course we can see that the condition $\lbrace g_{i},H\rbrace=b^{i1}g_{1}+\cdots+b^{in}g_{n}$ is more general than that defining a particular integral. More essentially, we can also observe that we have not taken $H$ as one of the functions $g_{1},\cdots,g_{n}$, it is because in the cosymplectic framework $H$ is not necessarily a particular integral (we may have $RH\neq 0$). 

Before presenting the proof of theorem \ref{tePI2}, we revise the Lie integrability by quadratures for a special kind of time-dependent smooth vector fields. Let $v$ be a smooth vector field over $\mathbb{R}^{n}\times\mathbb{R}$ where $v(x,t)=v^{1}(x,t)\frac{\partial}{\partial x^{1}}+v^{2}(x,t)\frac{\partial}{\partial x^{2}}+\cdots+v^{n}(x,t)\frac{\partial}{\partial x^{n}}+\frac{\partial}{\partial t}$ (observe that this is the local form of the evolution vector field $E_{H}$), with $t$ the coordinate in $\mathbb{R}$. The equations of motion of the dynamical system defined by the vector field $v$ read
\begin{equation}\label{eqt}
 \begin{array}{c}
\dot{x}^{1}=v^{1}(x^{1},x^2,\ldots,x^{n},t\,)\,, \\
\vdots \\
\dot{x}^{n}=v^{n}(x^{1},x^2,\ldots,x^{n},t\,)\,,\\
\dot{t} = \ 1 \ .
\end{array} 
\end{equation}
We trivially integrate the last equation to have the solution for the variable $t$. Thus, we solely need to solve a reduced system of $n$ ODE's. Explicitly
\begin{equation}\label{eqtreduced}
\begin{array}{c}
\dot{x}^{1}=v^{1}(x^{1},x^2,\ldots,x^{n},t\,)\,, \\
\vdots \\
\dot{x}^{n}=v^{n}(x^{1},x^2,\ldots,x^{n},t\,)\ ,
\end{array} 
\end{equation}
which can be solved by quadratures provided the existence of $n$ linearly independent vector fields of the form $u^{1}(x,t)\frac{\partial}{\partial x^{1}}+\cdots+u^{n}(x,t)\frac{\partial}{\partial x^{n}}$ that generate a solvable Lie algebra of symmetries for the vector field $\overline{v}=v^{1}(x,t)\frac{\partial}{\partial x^{1}}+v^{2}(x,t)\frac{\partial}{\partial x^{2}}+\cdots+v^{n}(x,t)\frac{\partial}{\partial x^{n}}$ \cite{AKN2006,CFGR2015,Cetal2016}.

The following result (well-known in the symplectic framework) is instrumental for the proof of theorem \ref{tePI2}.
\begin{lem}
\label{lem2}
If $g_{1},\cdots,g_{s}\in C^{\infty}(M)$ are functionally independent functions such that $Rg_{i}=0$ for $i=1,\cdots,s$ then the Hamiltonian vector fields $X_{g_{1}},\ldots,X_{g_{s}}$ are linearly independent.
\end{lem}
\begin{proof}
Let us suppose that $\alpha_{1}X_{g_{1}}+\cdots+\alpha_{s}X_{g_{s}}=0$ for some $\alpha_{1},\cdots,\alpha_{s}\in\mathbb{R}$. We have that $\left( \alpha_{1}X_{g_{1}}+\cdots+\alpha_{s}X_{g_{s}}\right) \lrcorner\Omega=\alpha_{1}dg_{1}+\cdots+\alpha_{s}dg_{s}$
then 
\begin{equation}
\alpha_{1}dg_{1}+\cdots+\alpha_{s}dg_{s}=0,
\end{equation}
which means that for regular points (points where the differential $dg_{1},\cdots,dg_{s}$ are linearly independent) the vector fields $X_{g_{1}},\ldots,X_{g_{k}}$ are linearly independent. Of course in the symplectic framework this result is obtained from the fact that $X_{g}\lrcorner\omega=dg$ for each $g\in C^{\infty}(M)$.
\end{proof}

Now we present the proof of theorem \ref{tePI2}.
\begin{proof}
From the fact that $g_{1},\cdots,g_{n}$ are functionally independent on $M$ it follows that $M_{g}=\lbrace x\in M: \ g_{1}(x)=\cdots=g_{n}(x)=0\rbrace$ defines a (smooth) submanifold of $M$ with dimension $dim(M)-n=2n+1-n=n+1$, and since $Rg_{1}=\cdots=Rg_{n}=0$, i.e., the functions $g_{1},\cdots,g_{n}$ are time-independent, then around any point in $M_{g}$ one can find local coordinates $(y^{1},\cdots,y^{n},t)$.

On the other hand we have that the evolution vector field is tangent to $M_{g}$, indeed, $E_{H}g_{i}=\lbrace g_{i},H\rbrace+Rg_{i}=b^{i1}g_{1}+\cdots+b^{in}g_{n}$ which is zero on $M_{g}$; so the trajectories of the Hamiltonian system $(M,\Omega,\eta,H)$ that live in $M_{g}$ are integral curves of $E_{H}|_{M_{g}}$ which in local coordinates $(y^{1},\cdots,y^{n},t)$ on $M_{g}$ has the form
\begin{equation}
E_{H}(y,t)=v^{1}(y,t)\frac{\partial}{\partial y^{1}}+\cdots+v^{n}(y,t)\frac{\partial}{\partial y^{n}}+\frac{\partial}{\partial t},
\end{equation} 
i.e., the trajectories of $(M,\Omega,\eta,H)$ that live in $M_{g}$ have the local form $\gamma(t)=(y^{1}(t),\cdots,y^{n}(t),t)$ where the functions $y^{1}(t),\cdots,y^{n}(t)$ are solutions of a system of differential equations of the form (\ref{eqtreduced}), namely,
\begin{equation}
\left\lbrace \begin{array}{c}
\dot{y}^{1}=v^{1}(y^{1},y^2,\ldots,y^{n},t), \\
\vdots \\
\dot{y}^{n}=v^{n}(y^{1},y^2,\ldots,y^{n},t),
\end{array} \right. ,
\end{equation} 
which can be solved by quadratures provided the existence of $n$ linearly independent vector fields of the form $u^{1}(y,t)\frac{\partial}{\partial y^{1}}+\cdots+u^{n}(y,t)\frac{\partial}{\partial y^{n}}$ that form a solvable Lie algebra of symmetries of the vector field $\overline{E}_{H}|_{M_{g}}=v^{1}(y,t)\frac{\partial}{\partial y^{1}}+\cdots+v^{n}(y,t)\frac{\partial}{\partial y^{n}}$. We have that the vector fields $X_{g_{1}},\ldots,X_{g_{n}}$ are tangent to the submanifold $M_{g}$ because $X_{g_{i}}g_{j}=\lbrace g_{j},g_{i}\rbrace$ which is zero on $M_{g}$, so we can consider the vector fields $X_{g_{1}}|_{M_{g}},\cdots,X_{g_{n}}|_{M_{g}}$ on $M_{g}$, which have the local form $u_{i}^{1}(y,t)\frac{\partial}{\partial y^{1}}+\cdots+u_{i}^{n}(y,t)\frac{\partial}{\partial y^{n}}$ and they generate a solvable Abelian Lie algebra with the Lie bracket of vector fields. Additionally, we have that $X_{g_{1}}|_{M_{g}},\cdots,X_{g_{n}}|_{M_{g}}$ are symmetries of $\overline{E}_{H}|_{M_{g}}$, indeed, we can observe that $\overline{E}_{H}|_{M_{g}}=X_{H}|_{M_{g}}$ so 
\begin{equation}
[\overline{E}_{H}|_{M_{g}},X_{g_{i}}|_{M_{g}}]=[X_{H}|_{M_{g}},X_{g_{i}}|_{M_{g}}]=X_{\lbrace g_{i},H\rbrace}|_{M_{g}}=0.
\end{equation}
We conclude that the vector fields $X_{g_{1}}|_{M_{g}},\cdots,X_{g_{n}}|_{M_{g}}$ on $M_{g}$, which have the local form $u_{i}^{1}(y,t)\frac{\partial}{\partial y^{1}}+\cdots+u_{i}^{n}(y,t)\frac{\partial}{\partial y^{n}}$, form a solvable Lie algebra of symmetries of the vector field $\overline{E}_{H}|_{M_{g}}$, therefore the trajectories of the Hamiltonian system $(M,\Omega,\eta,H)$ that live in $M_{g}$ can be found by quadratures.
\end{proof}
Following the previous theorem and the notion of particular integrability for the symplectic framework \cite{escobar2023particular}, we say that a cosymplectic Hamiltonian system, with $n$ degrees of freedom, possesses the property of particular integrability if it is possible to find functionally independent functions in particular involution. The number of these functions must be equal to $n$, and they do not dependent on time explicitly. As presented in \cite{escobar2023particular} for the symplectic case, we have that in the cosymplectic framework the
condition for particular integrability is also maximal, i.e., for a given cosymplectic Hamiltonian system, the largest number of time-independent functionally independent functions in particular involution coincides exactly with the number of degrees of freedom of the system.

\subsection{Contact Hamiltonian systems}
\label{subcontact}

Let $(M,\theta,H)$ be a contact Hamiltonian system with $n$ degrees of freedom (${\rm dim}(M)=2n+1$). Below, we introduce the notion of particular integral.
\begin{de}
A particular integral $g\in C^{\infty}(M)$ of $(M,\theta,H)$ is a function that obeys the equation $X_{H}g=ag$ for a function $a\in C^{\infty}(M)$.
\end{de}
We can observe that with this definition we still have the original characterization $\lbrace g,H\rbrace=ag$ (where in this case $\lbrace ,\rbrace$ defines a Jacobi bracket). Indeed, we know that $X_{H}g=\lbrace g,H\rbrace-gRH$, so $g\in C^{\infty}(M)$ is a particular integral if and only if $ag=\lbrace g,H\rbrace-gRH$ for some function $a\in C^{\infty}(M)$, if and only if $\lbrace g,H\rbrace=ag+gRH=bg$ with $b=a+RH\in C^{\infty}(M)$.

In \cite{BG2021,GGMRR2020,LL2020} the concept of dissipated quantities is viewed as a generalization of the notion of constants of motion (conserved quantities) for dissipative systems; namely, for the Hamiltonian system $(M,\theta,H)$ a function $g\in C^{\infty}(M)$ is called a dissipated quantity if $X_{H}g=-gRH$. 

The present notion of particular integral generalizes both concepts of constant of motion and dissipated quantity. Indeed, it is clear that a particular integral $g$ such that $X_{H}g=ag$ with $a=0$ is a constant of motion and with $a=-RH$ is a dissipated quantity. 

Let us assume that $g\in C^{\infty}(M)$ is a particular integral of $(M,\theta,H)$, we know that we can look for the trajectories of the system that live in $M_{g}=\lbrace\, x\in M: \ g(x)\,=\,0\,\rbrace $, which provided that $0$ is a regular value of $g$, is a smooth submanifold of dimension $2n+1-1=2n$. In canonical coordinates $(q^{1},\cdots,q^{n},p_{1},\cdots,p_{n},z)$ on $(M,\theta)$, the trajectories of the system are the solutions of the dissipative Hamilton's equations of motion 
\begin{equation}
\dot{q^{i}} =\frac{\partial H}{\partial p_{i}}, \hspace{1cm}
\dot{p_{i}} =-\left( \frac{\partial H}{\partial q^{i}}+p_{i}\frac{\partial H}{\partial z}\right) , \hspace{1cm}
\dot{z}=p_{i}\frac{\partial H}{\partial p_{i}}-H.
\end{equation}
Canonical transformations in contact mechanics are defined as transformations leaving the contact form invariant \cite{AE2023}, they are a particular case of the so called contact transformations, which are transformations leaving the contact form invariant up to a multiplicative conformal factor \cite{Arnold78,BCT2017}. A canonical transformation of the form $(q,p,z)\mapsto (Q,P,Z)$ such that $Q_{n}=g$ can be obtained by a generating function of the form $Z=z-F_{1}(q^{1},\cdots,q^{n},Q^{1},\cdots,Q^{n-1},g)$, where $F_{1}(q^{1},\cdots,q^{n},Q^{1},\cdots,Q^{n-1},g)$ is the generating function of a symplectic canonical transformation \cite{BCT2017}. If we can find a canonical transformation $(q,p,z)\mapsto (Q,P,Z)$ such that $P_{n}=g$, then in an analogous way as the reduction process in the symplectic and cosymplectic cases, we can find trajectories of $(M,\theta,H)$ that are trajectories of a reduced dynamical system where the dynamics is the projection of the original dynamics into a submanifold of the phase space in the same sense as in the symplectic case \cite{escobar2023particular}.

\begin{ex}
Let us investigate the motion of a point particle in a vertical plane under the influence of  constant gravity and the air friction. In canonical coordinates $(x,y,p_{x},p_{y},z)$ the Hamiltonian function is $H=\frac{1}{2m}(p_{x}^{2}+p_{y}^{2})+mgy+\gamma z$ \cite{GGMRR2020}. We have that $g=p_{x}$ is a dissipated quantity, therefore a particular integral, indeed $\lbrace p_{x},H\rbrace=0$. By taking $p_{x}=0$ we have that the system of Hamilton's equations reduces to 
\begin{equation}
\left\lbrace \begin{array}{c}
\dot{y}= \frac{\partial K}{\partial p_{y}}\\ 
\dot{p}_{y}= -\frac{\partial K}{\partial y}+p_{y}\frac{\partial K}{\partial z}\\
\dot{z}= p_{y}\frac{\partial K}{\partial p_{y}}-K
\end{array} \right.  \ ,
\end{equation}
and afterwards integrating the equation $\dot{x}=\frac{\partial H}{\partial p_{x}}$, where $K=H\big|_{p_{x}=0}=\frac{1}{2m}p_{y}^{2}+mgy+\gamma z$.
\end{ex}

Now we address the problem of introducing the study of particular integrability or the conditions form particular integrability in the framework of contact geometry. The notion of Liouville integrability in contact Hamiltonian mechanics differs from the one in the symplectic or cosymplectic framework; these differences are due to geometric aspects, in fact we have that symplectic and cosymplectic structures define Poisson structures on the phase space of a mechanical system, on the other hand, a contact structure defines a strictly Jacobi structure (these last structures are not common in the literature and are less known than the Poisson ones). Jacobi structures are more general than the Poisson ones, in fact brackets defined by Jacobi structures do not satisfy Leibniz rule, the actually fulfils the weak Leibniz
rule. Poisson and Jacobi structures were introduced by Lichnerowicz \cite{lichnerowicz1977varietes,lichnerowicz1978varietes}. In \cite{Boyer2011,Visinescu2017} the following definition is presented:
\begin{de}
A contact Hamiltonian system $(M,\theta,H)$ with $n$ degrees of freedom is completely integrable (in the Liouville sense) if there exist $n + 1$ constants of motion $H,g_{1},\cdots,g_{n}$ that are independent and in involution.
\end{de}
In this definition the functions $g_{1},\cdots,g_{k}$ are said to be independent if the corresponding Hamiltonian vector fields $X_{g_{1}},\cdots,X_{g_{k}}$ are linearly independent. This is not equivalent to the condition that the differentials $dg_{1},\cdots,dg_{k}$ are linearly independent, since the last condition is not satisfied when one of the Hamiltonian vector fields is $\mathbb{R}$-proportional to the Reeb vector field (the Hamiltonian function of the Reeb vector field is the constant $1$) \cite{Boyer2011}, i.e., under this context, the independence of functions is not equivalent to our notion of functional independence. Contact Hamiltonian systems with Hamiltonian function $H$ constant are called of Reeb type, in these cases the Hamiltonian vector field $X_{H}$ is $\mathbb{R}$-proportional to the Reeb vector field $R$ and the Hamilton's equations of motion are trivial, so we exclude this case in the notion on particular integrability. We have the following fundamental result relating functionally independent functions with linearly independent vector fields. 
\begin{lem}
\label{lem3}
If $g_{1},\cdots,g_{k}\in C^{\infty}(M)$ are functionally independent functions then the Hamiltonian vector fields $X_{g_{1}},\ldots,X_{g_{k}}$ are linearly independent.
\end{lem}
\begin{proof}
Let us suppose that $c_{1}X_{g_{1}}+\cdots+c_{k}X_{g_{k}}=0$ for some $c_{1},\cdots,c_{k}\in\mathbb{R}$. We know that 
\begin{equation}
\hspace{-1.0cm}\left( c_{1}X_{g_{1}}+\cdots+c_{k}X_{g_{k}}\right) \lrcorner d\theta=d(c_{1}g_{1}+\cdots+c_{k}g_{k})-R(c_{1}g_{1}+\cdots+c_{k}g_{k})\theta,
\end{equation}
then 
\begin{equation}
d(c_{1}g_{1}+\cdots+c_{k}g_{k})-R(c_{1}g_{1}+\cdots+c_{k}g_{k})\theta=0,
\end{equation}
or equivalently
\begin{equation}
d(c_{1}g_{1}+\cdots+c_{k}g_{k})=R(c_{1}g_{1}+\cdots+c_{k}g_{k})\theta.
\end{equation}
On the other hand, for any $f\in C^{\infty}(M)$ we have $f\theta\wedge d(f\theta)=f^{2}\theta\wedge d\theta$, then 
\begin{equation}
f\theta\wedge d(f\theta)^{n}=f^{n}\theta\wedge d\theta^{n}.
\end{equation}
So if $f\neq 0$ in an open subset of $M$ then $f\theta$ is a contact form on than open submanifold. Now if we suppose that $R(c_{1}g_{1}+\cdots+c_{k}g_{k})\neq 0$ in an open subset of $M$ then $R(c_{1}g_{1}+\cdots+c_{k}g_{k})\theta$ is a contact form on than open submanifold, but it cannot be since $d(R(c_{1}g_{1}+\cdots+c_{k}g_{k})\theta)=d(d(c_{1}g_{1}+\cdots+c_{k}g_{k}))=0$; so we must have that $R(c_{1}g_{1}+\cdots+c_{k}g_{k})$ is zero except possibly in a zero measure set, therefore $d(c_{1}g_{1}+\cdots+c_{k}g_{k})=0$ which implies that $c_{1}=\cdots=c_{k}=0$, i.e., for regular points (points where the differential $dg_{1},\cdots,dg_{k}$ are linearly independent) the vector fields $X_{g_{1}},\ldots,X_{g_{k}}$ are linearly independent.
\end{proof}

We can observe that in the definition of completely integrable contact Hamiltonian systems (in the Liouville sense) it is required the existence of $n+1$ constants of motion and, in addition, the Hamiltonian function is one of them. In general, the Hamiltonian function is not a constant of motion. In fact, contact Hamiltonian systems where the Hamiltonian function is a constant of motion (the Hamiltonian function is invariant under the flow of the Reeb vector field) are called good contact Hamiltonian systems; and completely integrable contact Hamiltonian system where the constants of motion are also invariant under the flow of the Reeb vector field are called completely good \cite{Boyer2011}. Even for the notion of non-commutative integrability in the contact case, it is required that the Hamiltonian function is a constant of motion \cite{Jovanovic2012,JJ2015}. The solutions of the Hamilton equations of motion of a completely good contact Hamiltonian system can be found by quadratures provided that the constants of motion form a solvable Lie algebra under the Jacobi bracket \cite{azuaje2024lie}.

We have the following theorem on integrability involving particular involution in the contact framework.
\begin{te}
\label{tePI3}
Let $(M,\theta,H)$ be a good contact Hamiltonian system. If $H=g_{1},\cdots,g_{n}\in C^{\infty}(M)$ are functionally independent and in particular involution functions such that $Rg_{i}=0$ then the trajectories of the system $(M,\theta,H)$ that live in $M_{g}=\lbrace x\in M: \ g_{1}(x)=\cdots=g_{n}(x)=0\rbrace$ can be found by quadratures. 
\end{te}
\begin{proof}
From the fact that $g_{1},\cdots,g_{n}$ are functionally independent functions on $M$, we have that $M_{g}=\lbrace x\in M: \ g_{1}(x)=\cdots=g_{n}(x)=0\rbrace$ is a smooth submanifold of $M$ of dimension $dim(M)-n=2n+1-n=n+1$; and since $Rg_{1}=\cdots=Rg_{n}=0$, i.e., the functions $g_{1},\cdots,g_{n}$ are $z$-independent, then around any point in $M_{g}$ we can find local coordinates $(y^{1},\cdots,y^{n},z)$.

On the other hand we have that the vector fields $X_{g_{1}},\cdots,X_{g_{n}}$ are tangent to $M_{g}$, indeed, $X_{g_{j}}g_{i}=\lbrace g_{i},g_{j}\rbrace-g_{i}Rg_{j}$ which is zero on $M_{g}$; so the trajectories of $(M,\theta,H)$ living in $M_{g}$ are integral curves of $X_{H=g_{1}}|_{M_{g}}$, i.e., the trajectories of $(M,\theta,H)$ that live in $M_{g}$ have the local form $\gamma(t)=(y^{1}(t),\cdots,y^{n}(t),z(t))$ and are solutions of the vector differential equation $\dot{x}=X_{H}(x)$ with $x=(y^{1},\ldots,y^{n},z)$ local coordinates on $M_{g}$, which provided the existence of $n+1$ linearly independent vector fields on $M_{g}$ that form a solvable Lie algebra of symmetries of $X_{H}(x)$, can be solved by quadratures.

In the contact framework we also have that the assignment $g\mapsto X_{g}$ defines a Lie algebra antihomomorphism between the Lie algebras $(C^{\infty}(M),\lbrace,\rbrace)$ and $(\mathfrak{X}(M),[,])$ \cite{LL2019}, indeed, for $h,g\in C^{\infty}(M)$ we have
\begin{equation}
-[X_{}h,X_{g}]=X_{\lbrace h,g\rbrace}.
\end{equation}
So we have that $X_{g_{1}}|_{M_{g}},\cdots,X_{g_{n}}|_{M_{g}}$ generate an Abelian (therefore solvable) Lie algebra of symmetries of $X_{H=g_{1}}|_{M_{g}}$ with the Lie bracket of vector fields, in addition the vector field $R$ is also tangent to $M_{g}$ and is a symmetry of $X_{H=g_{1}}|_{M_{g}}$, indeed, we know that $[R,X_{H}]=X_{\lbrace 1,H\rbrace}=0$. We conclude that the vector fields $R,X_{g_{1}}|_{M_{g}},\cdots,X_{g_{n}}|_{M_{g}}$ are linearly independent and generate a solvable (in fact it is an Abelian Lie algebra) Lie algebra of symmetries of $X_{H=g_{1}}|_{M_{g}}$ with the Lie bracket of vector fields, therefore, the trajectories of the Hamiltonian system $(M,\theta,H)$ that live in $M_{g}$ can be found by quadratures.
\end{proof}

\begin{de}
We say that a good contact Hamiltonian system $(M,\theta,H)$ with $n$ degrees of freedom satisfies the requirements for particular integrability if we can find $n$ functionally independent functions $H=g_{1},\cdots,g_{n}\in C^{\infty}(M)$ in particular involution such that $Rg_{i}=0$.
\end{de}

As in the symplectic and cosymplectic cases, we have that the condition for particular integrability is also a maximal condition in the contact case.

\subsection{Cocontact Hamiltonian systems}
\label{subcocontact}

We finish this section by introducing the notions of particular integral and particular integrability in the cocontact framework. The development of this subsection follows the guide of the previous ones. Let $(M,\theta,\eta,H)$ be a cocontact Hamiltonian system with $n$ degrees of freedom (${\rm dim}(M)=2n+2$). We introduce the notion of particular integral as follows.
\begin{de}
A particular integral of $(M,\theta,\eta,H)$ is a function $g\in C^{\infty}(M)$ which satisfies $R_{t}g=0$ and $X_{H}g=ag$ for a function $a\in C^{\infty}(M)$.
\end{de}
We can observe that, as in the previous subsection, in terms of the Jacobi bracket $g$ is a particular integral if and only if $\lbrace g, H\rbrace = bg$ for some function $b\in C^{\infty}(M)$. As in the cosymplectic case we require $R_{t}g=0$ in order to be able to look for the trajectories of the system that live in $M_{g}=\lbrace\, x\in M: \ g(x)\,=\,0\,\rbrace$, indeed, we must have $E_{H}g|_{g=0}=0$, where $E_{H}=X_{H}+R_{t}$ is the dynamical vector field. 

In \cite{GLR2022} dissipated quantities are studied for cocontact Hamiltonian systems.
\begin{de} A function $g\in C^{\infty}(M)$ is called a dissipated quantity of $(M,\theta,\eta,H)$ if it satisfies that $E_{H}g=-gR_{z}H$.
\end{de}
As it is remarked in \cite{BG2021,GLR2022}, in general, the Hamiltonian function $H$ obeys the relation $E_{H}H=-HR_{z}H+R_{t}H$. Thus, whenever $H$ depends explicitly on $t$ it is not a dissipated quantity itself. A dissipated quantity $g$ such that $R_{t}g=0$ is a particular integral.

Using an analogous procedure as in the contact framework, by means of a particular integral, we can obtain solutions of the equations of motion that are also solutions of a reduced system of ordinary differential equations, i.e., we can find trajectories of a cocontact Hamiltonian system that are trajectories of a reduced dynamical system.

As in the contact case, the notion of Liouville integrability for cocontact Hamiltonian systems is restricted to good Hamiltonian systems. A cocontact Hamiltonian system is good if the Hamiltonian function is invariant under the flow of the contact Reeb vector field; in this case the Hamiltonian function is not necessarily a constant of motion, indeed, let $(M,\theta,\eta,H)$ be a good cocontact Hamiltonian system, i.e., $R_{z}H=0$, then $E_{H}H=R_{t}H$ which is not necessarily zero. We say that a good cocontact Hamiltonian with $n$ degrees of freedom is a completely good cocontact Hamiltonian system if there exist $n$ functionally independent constants of motion $g_{1},\cdots,g_{n}$ with $R_{z}g_{i}=0$ that are in involution. The solutions of the Hamilton equations of motion of a completely good cocontact Hamiltonian system can be found by quadratures provided that the constants of motion form a solvable Lie algebra under the Jacobi bracket \cite{azuaje2024lie}. 

We have the following theorem analogous to theorems \ref{tePI2} and \ref{tePI3}.

\begin{te}
\label{tePI4}
Let $(M,\theta,\eta,H)$ be a good cocontact Hamiltonian system. Let $g_{1},\cdots,g_{n}\in C^{\infty}(M)$ be functionally independent and in particular involution functions satisfying $R_{z}g_{i}=R_{t}g_{i}=0$. If for each $i\in \lbrace 1,\cdots,n\rbrace$ $\lbrace g_{i},H\rbrace=c^{i1}g_{1}+\cdots+c^{in}g_{n}$ for some $c^{ij}\in C^{\infty}(M)$, then the trajectories of the Hamiltonian system $(M,\theta,\eta,H)$ living in $M_{g}=\lbrace x\in M: \ g_{1}(x)=\cdots=g_{n}(x)=0\rbrace$ can be found by quadratures. 
\end{te}
As in the cosymplectic case, we have not taken $H$ as one of the functions $g_{1},\cdots,g_{n}$ because $R_{t}H$ is not necessarily zero. We have the following result analogous to lemmas \ref{lem2} and \ref{lem3}. 
\begin{lem}
\label{lem4}
If $g_{1},\cdots,g_{k}\in C^{\infty}(M)$ are functionally independent functions satisfying $R_{t}g_{i}=0$ for $i=1,\cdots,k$ then the Hamiltonian vector fields $X_{g_{1}},\ldots,X_{g_{k}}$ are linearly independent.
\end{lem}
The proof is also analogous to the previous lemmas. Finally we show the proof of theorem \ref{tePI4}.
\begin{proof}
In this case we have that $M_{g}=\lbrace x\in M: \ g_{1}(x)=\cdots=g_{n}(x)=0\rbrace$ is again a (smooth) submanifold of $M$ of lower dimension $dim(M)-n=2n+2-n=n+2$, and since $R_{z}g_{1}=\cdots=R_{z}g_{n}=R_{t}g_{1}=\cdots=R_{t}g_{n}=0$, i.e., the functions $g_{1},\cdots,g_{n}$ are $z$-independent and time-independent, then around any point in $M_{g}$ we can find local coordinates $(t,y^{1},\cdots,y^{n},z)$.

On the other hand we have that the evolution vector field is tangent to $M_{g}$, indeed, $E_{H}g_{i}=\lbrace g_{i},H\rbrace+R_{t}g_{i}=c^{i1}g_{1}+\cdots+c^{in}g_{n}$ which is zero on $M_{g}$; so the trajectories of the Hamiltonian system $(M,\theta,\eta,H)$ that live in $M_{g}$ are integral curves of $E_{H}|_{M_{g}}$ which in local coordinates $(t,y^{1},\cdots,y^{n},z)$ on $M_{g}$ has the form
\begin{equation}
E_{H}(t,y,z)=u^{1}(t,y,z)\frac{\partial}{\partial y^{1}}+\cdots+u^{n}(t,y,z)\frac{\partial}{\partial y^{n}}+u^{z}(t,y,z)\frac{\partial}{\partial z}+\frac{\partial}{\partial t},
\end{equation} 
i.e., the trajectories of $(M,\theta,\eta,H)$ living in $M_{g}$ have the local form $\gamma(t)=(t,y^{1}(t),\cdots,y^{n}(t),z(t))$ where the functions $y^{1}(t),\cdots,y^{n}(t),z(t)$ are solutions of the system of differential equations
\begin{equation}
\left\lbrace \begin{array}{c}
\dot{y}^{1}=u^{1}, \\
\vdots \\
\dot{y}^{n}=u^{n}, \\
\dot{z}=u^{z},
\end{array} \right. 
\end{equation} 
which can be solved by quadratures provided the existence of $n+1$ linearly independent vector fields of the form $w^{1}(t,y,z)\frac{\partial}{\partial y^{1}}+\cdots+w^{n}(t,y,z)\frac{\partial}{\partial y^{n}}+w^{z}(t,y,z)\frac{\partial}{\partial z}$ that form a solvable Lie algebra of symmetries of the vector field $\overline{E}_{H}|_{M_{g}}=u^{1}(t,y,z)\frac{\partial}{\partial y^{1}}+\cdots+u^{n}(t,y,z)\frac{\partial}{\partial y^{n}}+u^{z}(t,y,z)\frac{\partial}{\partial z}$. As in the contact case we have that $R_{z},X_{g_{1}}|_{M_{g}},\cdots,X_{g_{n}}|_{M_{g}}$ are linearly independent and generate an Abelian Lie algebra of symmetries of $\overline{E}_{H}|_{M_{g}}=X_{H}|_{M_{g}}$ with the Lie bracket of vector fields. Therefore, the trajectories of  $(M,\theta,\eta,H)$ living in $M_{g}$ can be found by quadratures.
\end{proof}

\begin{de}
A good cocontact Hamiltonian system $(M,\theta,\eta,H)$, possessing $n$ degrees of freedom, satisfies the requirements for particular integrability if we can find $n$ functions $g_{1},\cdots,g_{n}\in C^{\infty}(M)$ functionally independent and in particular involution satisfying $R_{z}g_{i}=R_{t}g_{i}=0$ for each $i\in \lbrace 1,\cdots,n\rbrace$ and  $\lbrace g_{i},H\rbrace=c^{i1}g_{1}+\cdots+c^{in}g_{n}$ for some $c^{ij}\in C^{\infty}(M)$.
\end{de}

Of course, like in the contact case, the condition for particular integrability in the cocontact framework is a maximal condition.

\section{Conclusions}

We have introduced the notions of particular integral and particular integrability for classical Hamiltonian systems on cosymplectic, contact and cocontact manifolds. These rigorous results extend and complement those presented in \cite{escobar2023particular} for the symplectic case. A constant of motion (a first Liouville integral) can be regarded as a special case of the concept of particular integral. Especially, in the contact framework the present study generalizes the important notion of a dissipated quantity. For a non necessarily integrable (in the Liouville sense) cosymplectic, contact or cocontact Hamiltonian system we have shown that, provided the existence of a sufficient number of functionally independent particular integrals, it is possible to find some of the trajectories by quadratures. If we do not posses enough number of particular integrals required to achieve particular integrability, we are still able to reduce the Hamilton's equations of motion to a simpler form. 

For future work, we plan to investigate the existence of infinitesimal symmetries related to particular integrals in the above mentioned geometric frameworks. Particular superintegrability (more particular integrals than degrees of freedom) from a geometric perspective is also an interesting open question, (for an algebraic novel geometric framework for conformally superintegrable systems see, for instance, \cite{Kress}).

\section*{Acknowledgments}

The author R Azuaje wishes to thank CONAHCYT (México) for the financial support through a postdoctoral fellowship in the program Estancias Posdoctorales por México 2022. In addition, R Azuaje thanks Alessandro Bravetti for his helpfull comments on some aspects of contact geometry.

A M Escobar Ruiz would like to thank the support from Consejo Nacional de Humanidades, Ciencias y Tecnologías (CONAHCyT) of Mexico under Grant CF-2023-I-1496 and from UAM research grant 2024-CPIR-0.

The authors thank Ond\v{r}ej Kub\u{u} for some relevant comments and helpful discussions.

\bibliography{refs} 
\bibliographystyle{unsrt} 

\end{document}